\documentclass[11pt]{article}

\usepackage[margin=1in]{geometry}
\usepackage{amsmath,amssymb,amsthm,mathtools}
\usepackage{enumitem}
\usepackage{xcolor}
\usepackage[colorlinks=true,linkcolor=blue,citecolor=blue,urlcolor=blue,
  pdftitle={Local Search for Fair Max--Min Diversification},
  pdfauthor={Sepideh Mahabadi, Shyam Narayanan, and Varun Sivashankar}]{hyperref}

\usepackage{todonotes}

\newtheorem{theorem}{Theorem}[section]

\newtheorem{lemma}[theorem]{Lemma}

\theoremstyle{definition}
\newtheorem{definition}[theorem]{Definition}
\theoremstyle{remark}

\DeclareMathOperator{\diver}{div}
\DeclareMathOperator{\Bad}{Bad}
\DeclareMathOperator{\poly}{poly}
\newcommand{\col}{\operatorname{col}}
\newcommand{\dist}{d}

\newcommand{\Z}{\mathbb{Z}}

\title{Local Search for Fair Max-Min Diversification}
\author{
Sepideh Mahabadi\\
Microsoft Research\\
\texttt{smahabadi@microsoft.com}
\and
Shyam Narayanan\footnote{Work done while a student at MIT}\\
Citadel Securities\\
\texttt{shyam.s.narayanan@gmail.com}
\and
Varun Sivashankar\\
Princeton University\\
\texttt{varunsiva@princeton.edu}
}
\date{}

\begin{document}

\maketitle

\begin{abstract}
Given $n$ points in a metric space, Max-Min diversification asks for a subset of
$k$ points maximizing the minimum pairwise distance between the selected points. This is arguably the most fundamental notion of diversity with applications across a wide range of domains.
We consider this problem under partition constraints, previously studied as Fair
Max-Min Diversification (FMMD). Here, each point has a color in $[m]$, and a feasible solution must contain exactly $k_i$ points of color $i$, where $k_1,\ldots,k_m$ are prescribed parameters satisfying $\sum_i k_i=k$.
We give the first constant factor approximation for the problem using local search, that runs in time $f(m)\cdot \poly(n)$, in which
all constraints are satisfied exactly.
All previously known algorithms either provided an $\widetilde \Theta(m)$ approximation factor, had running times exponential in the solution size $k$, or satisfied the fairness
constraints only approximately or in expectation.

We further generalize our result to the problem where each point may belong to an arbitrary subset of colors. Given lower and upper bounds $\ell_i$ and $u_i$ for every color $i$, the goal is to find $k$ points whose color counts satisfy all these bounds while maximizing their diversity.
\end{abstract}

\section{Introduction}

Selecting a small but diverse subset is a basic primitive in data
summarization, search, recommendation, and facility location
~\cite{DrosouJPS2017,QinYC2012,StoyanovichYJ2018,CelisKSDKV2018}.
Given $n$ points in a metric space, \emph{Max-Min Diversification}
asks for a subset of $k$ points maximizing the minimum pairwise
distance among the selected points ~\cite{ChandraH2001,RaviRT1994,Kuby1987,Erkut1990}. This is arguably the most
fundamental notion of diversity: rather than optimizing an aggregate
statistic, it requires every pair of selected points not to be too similar.

Max-Min Diversification is one of several distance-based diversity
objectives studied in the framework of Chandra and
Halld\'orsson~\cite{ChandraH2001} (denoted as {\em remote-edge}).
The two other practical notions of diversity are the sum of pairwise distances
(denoted as {\em Sum-Sum} or {\em remote-clique}) and the sum, over
the selected points, of the distance from each point to its nearest
selected neighbor ({\em Sum-Min} or {\em remote-pseudoforest})
~\cite{ChandraH2001,bhaskara2016linear}.
Even without additional constraints, these
problems are NP-hard.  For Max--Min Diversification, the classical
greedy algorithm of Gonzalez gives a \(2\)-approximation in
general metrics, and this factor is best possible in general metrics
unless \(P=NP\)
~\cite{Gonzalez1985,RaviRT1994}.

In many applications, the selected set must additionally satisfy prescribed budget requirements. This is used to capture different categories (e.g. in a recommendation system) or populations (e.g. for fairness purposes), or to capture other requirements such as recency in a summary \cite{MahabadiT2023,MoumoulidouMM2021}. This requirement is modeled by partitioning
the points into \(m\) color classes and prescribing a quota \(k_i\) for
each class.  The goal is to select exactly $k_i$ points of color
\(i\), while maximizing the minimum pairwise distance.  This problem,
known as \emph{Fair Max-Min Diversification} (FMMD) ~\cite{MoumoulidouMM2021}, is equivalently
Max-Min Diversification over the bases of a partition matroid.

For the two other diversity objectives above, constant-factor
polynomial-time approximations are known under matroid constraints:
local search gives a $2$-approximation for remote-clique, and an
LP-rounding approach gives a randomized $8$-approximation for
remote-pseudoforest, even under general matroid constraints
~\cite{AbbassiMT2013,bhaskara2016linear}.  The situation for Max-Min is
strikingly different. At a high level, this difference stems from the fragility of the Max-Min objective compared with the other two: its value is determined by a single closest pair. This fragility, combined with the requirement of exact color quotas, couples choices and requires coordination across groups. Indeed, for arbitrary $m$, all previously known algorithms for
general metrics incurred at least one of three limitations: either their
approximation factor grew near-linearly with $m$, or their running time was
exponential in the solution size \(k\), or they satisfied the fairness
constraints only approximately or in expectation
~\cite{MoumoulidouMM2021,AddankiMMM2022,WangMLF2023,AdriaensT2025}.

In particular, it was not known whether FMMD admits a constant-factor approximation that satisfies all quotas exactly and runs in $f(m)\poly(n)$ time. We resolve this question affirmatively, with $f$ doubly exponential in $m$. We further extend our result to overlapping colors with lower and upper quotas, as formalized below.

\subsection{Problem Formulation}
\label{sec:problem-formulation}

For a positive integer \(m\), write \([m]=\{1,\ldots,m\}\).  Let \((\mathcal{X},\dist)\) be a
metric space and let $X\subseteq \mathcal{X}$ be a finite point set given as input to our problem with \(|X|=n\).  We define the {\em diversity} of a subset \(T\subseteq X\) as
\[
        \diver(T)=\min_{\substack{p,q\in T\\p\ne q}}\dist(p,q).
\]
For \(|T|\le1\), we set \(\diver(T)=\infty\).  In either model below,
\(\mathrm{OPT}\) denotes the maximum diversity of a feasible set.

\paragraph{Exact quotas.} In the \emph{exact-quota FMMD problem}, the ground set $X$ is partitioned into color classes
\(X_1,\ldots,X_m\), and the input specifies nonnegative quotas
\(k_1,\ldots,k_m\), where \(k_i\le |X_i|\).  A set \(T\subseteq X\) is feasible if $|T\cap X_i|=k_i$ for all $i\in[m]$. We let $k=\sum_i k_i$ denote the solution size.

\paragraph{General case.} In the \emph{general case}, a point \(p\) may belong to several colors.  Its profile
\(\sigma(p)\subseteq[m]\) is the set of colors to which it belongs.  The input specifies a
cardinality \(k\), and integer bounds \(0\le \ell_i\le u_i\le k\) for $i\in [m]$.  For
\(T\subseteq X\), let
\[
        g(T)_i=|\{p\in T:i\in\sigma(p)\}|.
\]
Writing \(\ell=(\ell_1,\ldots,\ell_m)\) and \(u=(u_1,\ldots,u_m)\), a set \(T\) is
feasible if $|T|=k$ and $\ell\le g(T)\le u$, where here, as well as throughout the paper, vector inequalities are interpreted coordinatewise.  Exact quotas are the special case in which
every profile is a singleton, \(\ell_i=u_i=k_i\), and \(k=\sum_i k_i\).  

\subsection{Our Results}
In this work, we prove the following two results.
\begin{enumerate}[label=(\roman*),leftmargin=*]
    \item \emph{Exact quotas (Theorem~\ref{thm:main-approx}).}
    There is a \(6\)-approximation for exact-quota FMMD with running time
    \(2^{O(m2^m)}\poly(n)\).

    \item \emph{General case (Theorem~\ref{thm:general-approx}).}
    There is a \(6\)-approximation with running time
    \(2^{2^{O(m^2)}}\poly(n)\).
\end{enumerate}

Both algorithms satisfy all fairness constraints exactly, and their approximation factor
is independent of \(m\).  The non-polynomial dependence is only on the number of classes,
not on the solution size \(k\).  Exact quotas are a special case of the general model, but their additional structure yields a better running-time dependence on $m$ in~(i).

\subsection{Prior and Related Work}
Moumoulidou et al.~\cite{MoumoulidouMM2021} introduced fair max-min diversification
with disjoint exact quotas in general metrics.  They gave a
factor \(4\)-approximation for $m=2$, and a \((3m-1)\)-approximation for arbitrary $m$.
They also gave a \(5\)-approximation whose running time is exponential in the solution
size \(k\).  Addanki et al.~\cite{AddankiMMM2022} improved the approximation factor from $3m-1$ to \(m+1\).  They further gave constant factor approximation algorithms by relaxing the fairness constraints to hold only in expectation or up to
multiplicative $(1-\epsilon)$ factors.  They also gave algorithms for one-dimensional, Euclidean, and
bounded-doubling metrics with better approximation factor but with similar limitations on the runtime or fairness guarantees.  Kurkure et al.~\cite{KurkureSWGS2024} developed
faster algorithms for the Euclidean case, again providing only approximate
fairness guarantees.

\medskip
For disjoint groups with lower and upper quotas and a fixed solution size, Wang et
al.~\cite{WangMLF2023} obtained diversity at least
\((1-\varepsilon)\mathrm{OPT}/5\), but with a running time containing \(m^k\).
Adriaens and Tatti~\cite{AdriaensT2025} gave a randomized polynomial-time algorithm
with an improved approximation factor of $O(m/\sqrt{\log m})$ when \(k\le m\), which translates to $O(k/\sqrt{\log k})$ when $k\geq m$.

\medskip

Moumoulidou et al.~\cite{MoumoulidouMM2021} also considered overlapping groups.  Their
model has only lower quotas, with neither a fixed solution size nor upper quotas.  They
obtained approximation factor \(3\binom{m}{\lfloor m/2\rfloor}-1\).  Our general result
allows overlapping colors, lower and upper quotas, and a fixed solution size, while
satisfying every constraint exactly.  

\medskip

The problem has also been studied in massive-data
models, including streaming algorithms with exact color requirements
~\cite{WangFM2022,WangFML2023} and coreset constructions for fair
diversity objectives~\cite{MahabadiT2023}. Other work considers
variants in which the metric axioms are relaxed~\cite{GaoB2024}.

\subsection{Technical Overview}
We begin with the exact-quota problem. Assume that we have a correct guess of \(\mathrm{OPT}\). The algorithm tries all
pairwise distances, so one of these guesses is correct.

\paragraph{Preprocessing.} As a
preprocessing step, it is easy to show that one can sparsify each color class so that any two points
of the same color are at distance at least \(\mathrm{OPT}/3\).  This
loses at most a factor of \(3\). Thus if we set $r=\frac{\mathrm{OPT}}{6}$, we can assume that the points of the same color are at distance at least \(2r\), and that the optimal solution \(S\) has diversity greater than \(2r\). We then aim for diversity $r$ which would give us a factor $6$ approximation overall.
By the sparsification step, different candidates of the same color are far enough, so a point of the current solution can block at most one candidate of any fixed color.

\paragraph{The local-search potential.}
Local search has been used successfully for several diversity
maximization problems~\cite{ChandraH2001,AbbassiMT2013}. A main
question in our setting is which potential function to optimize. One can consider the 
minimum pairwise distance in the current solution $T$, which might require $\Omega(k)$ simultaneous swaps to improve. Aggregate quantities, such as the sum of pairwise distances on the other hand, do not
directly capture whether all pairs are sufficiently separated.

It turns out that the function which works is the number of {\em bad points} in $T$, i.e., points that have at least another point in $T$ within distance $r$ of them.
If there are no bad points, the current solution already
has diversity at least \(r\).  
So the goal is to show that as long as there is a bad point in $T$, there exists a set of swaps of size $h=h(m)$ that is quota-preserving, that decreases the number of bad points.  Thus the process terminates after at most \(k\leq n\) repairs.
This would yield a local-search algorithm running in
\(k\cdot n^{h(m)}\) time. One can then combine it with the coreset idea
of~\cite{MoumoulidouMM2021} to obtain a running time of
\((km)^{h(m)+1}\), at the cost of an additional constant-factor loss in
the approximation.

\paragraph{A repair of bounded size exists.}
Consider a pair of bad points in \(T\). Since \(S\) has diversity greater than \(2r\), at least one endpoint, say \(x\), does not belong to \(S\). Suppose, for intuition, that \(x\) is green. Since \(S\) and \(T\) satisfy the same quotas, \(S\setminus T\) contains an unused green candidate that can replace \(x\). This candidate may be too close to several points of \(T\), which we call its \emph{blockers}. Removing a blocker creates a vacancy in its color, which must be filled by another candidate from \(S\setminus T\); that candidate may have blockers of its own, so the repair can branch recursively across colors.

Although this recursive process may at first appear unbounded, we show that there always exists a repair whose dependencies form a rooted tree. Each candidate has at most one blocker of each color, so every node has at most \(m\) children. Moreover, using a minimality argument, we show that no color appears twice on a root-to-leaf path, and hence the tree has height at most \(m\). This gives a repair of size at most \(h(m)=m^m\), and a sharper analysis improves the bound to \(h(m)=2^{m}\).
The candidates in this tree can be inserted simultaneously while removing $x$ and the blockers of the tree nodes. This preserves all quotas, creates no new bad points, and eliminates $x$. Thus the number of bad points strictly decreases.

\paragraph{Sparsifying the candidate set using their types.}
Naively, enumerating all repairs of size at most \(h(m)\) gives a running time of \(n^{O(h(m))}\). To make the exponent of \(n\) independent of \(m\), we sparsify the set of candidate points. We assign each candidate a {\em type}, consisting of its color and the set of colors of its blockers. It is easy to see that there are at most \(m2^{m}\) types. We then define a conflict graph in which two candidates conflict if they are too close or share a blocker. The preprocessing step implies that this graph has degree bounded by a function of \(m\), and that candidates of the same type do not conflict. We can therefore retain only \(h'(m)\) candidates of each type independently, while preserving the existence of a bounded repair. Indeed, any repair of size at most \(h(m)\) can be transformed, one type at a time, into a repair using only the retained candidates. The resulting kernel has size depending only on \(m\), and exhaustive search on this kernel gives an \(f(m)\poly(n)\)-time algorithm.

\paragraph{The general case.}

When points may carry several colors and the constraints are given by lower and upper quotas, the dependency-tree argument no longer applies because adding or removing one point may affect several quotas simultaneously. Instead, we represent each candidate and its blockers by an integer vector recording their net effect on the cardinality and all color counts. We also allow additional deletion moves. Comparing the current solution with a well-separated feasible solution gives a possibly large valid repair. Using a sign-compatible decomposition based on Steinitz's lemma, we extract from it a repair of size \(2^{O(m^2)}\) that still satisfies all lower and upper quotas. The type-based sparsification argument then extends to these vector-valued moves, giving the running time \(2^{2^{O(m^2)}}\poly(n)\).

\section{Preliminaries}
\label{sec:preliminaries}

\subsection{Distance Reduction}
\label{sec:common-sparsification}

Recall that the profile \(\sigma(p)\subseteq[m]\) of a point \(p\) is the set of its colors.
Applying the greedy sparsification of Moumoulidou et
al.~\cite[proof of Theorem~6]{MoumoulidouMM2021} without truncation, we may assume that
any two points with the same profile are at distance at least \(\mathrm{OPT}/3\).  The same
untruncated reduction appears in Adriaens and
Tatti~\cite[Proposition~2]{AdriaensT2025}.

\begin{lemma}
\label{lem:profile-sparsification}
Given the guessed value \(\mathrm{OPT}>0\), one can find in polynomial time a set
\(Y\subseteq X\) such that:
\begin{enumerate}[label=(\roman*)]
    \item distinct points of \(Y\) with the same profile are at distance at least
    \(\mathrm{OPT}/3\).
    \item if \(X\) contains a feasible set of diversity at least \(\mathrm{OPT}\), then
    \(Y\) contains a feasible set of diversity greater than \(\mathrm{OPT}/3\).
\end{enumerate}
\end{lemma}

\begin{proof}
For each profile \(P\) occurring in \(X\), let
\(X_P=\{p\in X:\sigma(p)=P\}\).  Greedily choose a maximal set
\(Y_P\subseteq X_P\) whose points are pairwise at distance at least \(\mathrm{OPT}/3\), and let
\(Y=\bigcup_PY_P\).  Maximality implies that every point of \(X_P\) is within distance
\(<\mathrm{OPT}/3\) of some point of \(Y_P\).

Let \(O\) be a feasible set of diversity at least \(\mathrm{OPT}\).  Map each point of
\(O\cap X_P\) to a nearest point of \(Y_P\).  This map is one-to-one within each profile,
since two points with the same image would be at distance less than \(2\mathrm{OPT}/3\).  Images
of points with different profiles are also distinct.  Every point is replaced by one with
the same profile, so the image has the same cardinality and color counts as \(O\) and is
therefore feasible.
For two distinct image points \(y,y'\), with preimages \(o,o'\), the triangle inequality
gives
\(
\dist(y,y')\ge \dist(o,o')-\dist(o,y)-\dist(o',y')
>\mathrm{OPT}-\mathrm{OPT}/3-\mathrm{OPT}/3=\mathrm{OPT}/3
\).
Thus the image has diversity greater than \(\mathrm{OPT}/3\).
\end{proof}

The algorithm tries every pairwise distance as the value of \(\mathrm{OPT}\).  For each
guess, apply Lemma~\ref{lem:profile-sparsification} and set \(r=\mathrm{OPT}/6\).  For the
correct guess, distinct points with the same profile are at distance at least \(2r\), and
the reduced instance has a feasible solution of diversity greater than \(2r\).  We return
the best solution found over all guesses.

\subsection{Local-search Potential}

For a feasible set \(T\) and a threshold \(r>0\), define
\[
        \Bad_r(T)=\{p\in T:\exists q\in T\setminus\{p\}\text{ with }\dist(p,q)<r\}.
\]
When \(r\) is fixed, we write simply \(\Bad(T)\).  A feasible set has diversity at least
\(r\) if and only if \(\Bad_r(T)=\varnothing\).
We call a feasible set \(T\) \emph{bad} if \(\Bad(T)\ne\varnothing\).
The local swaps constructed below preserve feasibility and, for some
\(x\in\Bad(T)\), produce a set \(T'\) satisfying
\[
        \Bad(T')\subseteq\Bad(T)\setminus\{x\}.
\]
Thus each swap strictly decreases \(|\Bad(T)|\), and at most \(n\) swaps are needed to
obtain a feasible set of diversity at least \(r\).

\section{Exact Quotas}
\label{sec:exact-quotas}

In this section, each point \(p\in X\) has a color \(\col(p)\in[m]\).  For a finite set
\(A\subseteq X\), let \(\chi(A)\in\Z_{\ge0}^m\) be its color-count vector, so
\(\chi(A)_i=|A\cap X_i|\).
Let \(e_i\) denote the \(i\)th standard basis vector.

\begin{theorem}[Exact-quota FMMD]
\label{thm:main-approx}
There is a deterministic algorithm for FMMD in general metrics that returns a feasible
solution \(T\) with \(\diver(T)\ge \mathrm{OPT}/6\) in time
\(2^{O(m2^m)}\poly(n)\).
\end{theorem}

Let \(\mathrm{OPT}\) denote the optimum value of the original instance.
Apply Lemma~\ref{lem:profile-sparsification}, and set
\(r=\mathrm{OPT}/6\).  For notational simplicity, we continue to write
\(X\) for the reduced ground set.  Distinct points of the same color
in \(X\) are at distance at least \(\mathrm{OPT}/3=2r\), and the
reduced instance contains a feasible set \(S\) with
\(\diver(S)>\mathrm{OPT}/3=2r\).  We fix such a set \(S\) throughout
this section.

\subsection{Certificates}
\label{sec:local-certificates}

Let \(T\) be any feasible solution.  If \(\Bad(T)=\varnothing\), then \(T\) already has
diversity at least \(r\).  Suppose instead that \(T\) is bad.  We repair it by removing a
bad point \(x\) and adding new points.  Every old point within distance less than \(r\) of
an addition must also be removed, and the additions must replace all these removed points
color by color.

Fix a bad point \(x\in \Bad(T)\) and let \(a=\col(x)\).  For a candidate point
\(y\in X\setminus T\), define its blocker set, excluding \(x\), by
\[
        B_x(y)=\{p\in T\setminus\{x\}:\dist(p,y)<r\}.
\]
For \(A\subseteq X\setminus T\), write
\[
        B_x(A)=\bigcup_{y\in A}B_x(y).
\]
We call \(A\) \emph{compatible} if its points are pairwise at distance at least \(r\) and
the blocker sets \(B_x(y)\), \(y\in A\), are pairwise disjoint. 

\begin{definition}[Certificate]
\label{def:certificate}
A certificate for \(x\) is a compatible set \(A\subseteq X\setminus T\) such that, for
every color other than \(a\), the additions and their blockers contain the same number of
points of that color, while for color \(a\), the additions contain one extra point.
Equivalently,
\[
        \chi(A)-\chi(B_x(A))=e_a.
\]
\end{definition}

\begin{lemma}
\label{lem:certificate-improves}
If \(A\) is a certificate for \(x\in\Bad(T)\), then
\[
        T' = T\setminus\bigl(B_x(A)\cup\{x\}\bigr)\cup A
\]
is feasible and satisfies
\[
        |\Bad(T')|<|\Bad(T)|.
\]
\end{lemma}

\begin{proof}
The color-count condition gives
\[
        \chi(A)=\chi(B_x(A))+e_a.
\]
Since \(x\) has color \(a\), the removed set \(B_x(A)\cup\{x\}\) has the same color-count
vector as \(A\).  Hence \(T'\) satisfies every quota exactly.

The points of \(A\) are pairwise at distance at least \(r\) because \(A\) is compatible.
Every point of \(T\setminus(B_x(A)\cup\{x\})\) is also at distance at least \(r\) from
every point of \(A\), since all closer points belong to \(B_x(A)\) and were removed.  The
pairs among the old points that remain are unchanged.  Therefore
\[
        \Bad(T')\subseteq \Bad(T)\setminus\{x\}.
\]
Since \(x\in\Bad(T)\), this containment gives \(|\Bad(T')|<|\Bad(T)|\).
\end{proof}

\subsection{A Bounded Certificate}
\label{sec:hidden-certificate}

\begin{lemma}
\label{lem:hidden-certificate}
If \(T\) is bad, then some bad point \(x\in\Bad(T)\) has a certificate \(A^\star\) with
\(|A^\star|\le 2^{m-1}\).
\end{lemma}

\begin{proof}
Pick a bad edge \(\{u,v\}\subseteq T\).  Since \(\diver(S)>2r\), \(u\) and \(v\) cannot
both belong to \(S\).  Choose \(x\in\{u,v\}\setminus S\), and put
\(a=\col(x)\) and \(S^+=S\setminus T\). 

The sets \(B_x(s)\), \(s\in S^+\), are disjoint, since a common blocker would put two
points of \(S\) less than \(2r\) apart.  Each contains at most one point of each color,
since same-color points are \(2r\)-separated.  Moreover,
\(B_x(S^+)\subseteq T\setminus S\): a blocker in \(T\cap S\) would be too close to another
point of \(S\).

\paragraph{A minimal replacement.}
We have \(B_x(S^+)\cup\{x\}\subseteq T\setminus S\), so
\[
        \chi(B_x(S^+))+e_a
        \le\chi(T\setminus S).
\]
Since \(T\) and \(S\) contain the same number of points of every color,
\[
        \chi(T\setminus S)=\chi(S\setminus T)=\chi(S^+).
\]

Choose an inclusion-minimal \(D\subseteq S^+\) satisfying this inequality, i.e., $\chi(B_x(D))+e_a\le\chi(D)$. Note that $D$ cannot be the $\emptyset$ as we should have $\chi(D)\geq e_a$.  In fact we show that the inequality is tight:
\[
        \chi(B_x(D))+e_a=\chi(D).
\]
If it were strict in color \(i\), choose \(s\in D\) of color \(i\).  Since the blocker sets
are disjoint,
\[
        B_x(D\setminus\{s\})=B_x(D)\setminus B_x(s).
\]
Removing \(s\) decreases the right side by \(e_i\) and subtracts \(\chi(B_x(s))\) from the
left, so the inequality still holds, contradicting the minimality of \(D\).

\paragraph{The dependency tree.}
Thus, for each color other than \(a\), we can match the additions in \(D\) one-to-one with the blockers
of that color.  For color \(a\), we match all but one addition, and denote the unmatched
addition by \(s_0\).  Consider a directed graph $G$ on the set of additions in $D$, where for every blocker of an addition \(s\), draw an edge from \(s\) to
the addition matched to that blocker.  Thus \(s_0\) has indegree zero and every other
vertex has indegree one.

Define another directed (multi-)graph $G'$ on the set of colors by putting an edge \(i\to j\) whenever some addition of
color \(i\) has a blocker matched to an addition of color \(j\), i.e., for any edge $(a,b)$ in $G$ where $a$ has color $i$ and $b$ has color $j$ we put an edge $(i,j)$ in $G'$.  This color graph is
acyclic: Otherwise, consider a directed
cycle in $G'$ like \(i_1\to i_2\to\cdots\to i_t\to i_1\).  For each edge choose a point
\(s_j\in D\) of color \(i_j\) with a blocker of color \(i_{j+1}\), and let
\(D_0=\{s_1,\ldots,s_t\}\).  The chosen blockers have the same colors as the points of
\(D_0\), so
\[
        \chi(B_x(D_0))\ge\chi(D_0).
\]
Since the blocker sets are disjoint,
\[
        \chi(B_x(D\setminus D_0))
        =\chi(B_x(D))-\chi(B_x(D_0)).
\]
Using \(\chi(B_x(D))+e_a=\chi(D)\) and
\(\chi(B_x(D_0))\ge\chi(D_0)\), we obtain
\[
        \chi(B_x(D\setminus D_0))+e_a
        \le \chi(D)-\chi(D_0)
        =\chi(D\setminus D_0).
\]
and again note that $D\setminus D_0\neq \emptyset$. This contradicts the minimality of \(D\).

A directed cycle in the graph $G$ on \(D\) would give a directed closed walk in the color graph
$G'$ and hence a directed cycle there.  Thus the graph $G$ on \(D\) is also acyclic.  Since \(s_0\) is its
only vertex of indegree zero and every other vertex has indegree one, it is a tree rooted at
\(s_0\). Note that $G'$ is acyclic and thus it naturally gives us a topological ordering $\mathcal{O}$ on the colors.   Now consider the graph $G$. Our goal is to bound the number of vertices in $G$ (i.e., $|D|$) which is equal to the number of root-to-vertex paths in $G$. Note that along every root-to-vertex path in $G$,
the colors are distinct and should only appear according to the ordering $\mathcal{O}$.  Thus, given the set of colors on the path determines their sequence as well. Moreover, each vertex has at most one child of any
given color because otherwise its blocker set would have two points of the same color which is not possible by Lemma~\ref{lem:profile-sparsification}, so this sequence identifies at most one vertex in $G$.  There are only \(2^{m-1}\) subsets
of colors excluding \(a\), and hence $|D| \le 2^{m-1}$.

Finally, for \(D\subseteq S\setminus T\), its points are more than \(2r\) apart, and its blocker
sets are disjoint.  Thus \(D\) is compatible.  The tight equality says
\(\chi(D)-\chi(B_x(D))=e_a\), so \(A^\star=D\) is a certificate for \(x\) of the claimed
size.
\end{proof}
\subsection{Sparsifying the Candidate Set and Finding a Certificate}
\label{sec:exact-kernel}

\begin{lemma}
\label{lem:find-exact-certificate}
Suppose distinct points of the same color are at distance at least \(2r\), and the instance
has an optimal solution of diversity greater than \(2r\).  Given a bad feasible set \(T\),
one can find a bad point \(x\) and a certificate for \(x\) in time
\(2^{O(m2^m)}\poly(n)\).
Without these assumptions, the same search may report \(\mathrm{FAIL}\).
\end{lemma}

\begin{proof}
For each \(x\in\Bad(T)\), define a graph \(H_x\) on \(X\setminus T\), adding an edge between candidates
\(y,z\) when \(\dist(y,z)<r\) or \(B_x(y)\cap B_x(z)\ne\varnothing\).  Its independent
sets are exactly the compatible candidate sets.

For each candidate \(y\) define its {\em label} as
\[
        \lambda_x(y)=\bigl(\col(y),\{\col(p):p\in B_x(y)\}\bigr).
\]
Because \(B_x(y)\) contains at most one point of each color, this label determines
both \(\chi(\{y\})\) and \(\chi(B_x(y))\).  Thus replacing a candidate by one with the
same label preserves its contribution to the certificate equation.
There are at most \(m2^{m-1}\) different labels.  Candidates with the same label are independent:
they have the same color, and a shared blocker would place two of them at distance less
than \(2r\).  Moreover, the maximum degree \(\Delta=\Delta(H_x)\) is at most
\((m-1)^2\).  Indeed, a candidate \(y\) has at most one neighbor within distance $r$ of each other
color, and hence at most \(m-1\) in total.  Further, it has at most \(m-1\) blockers, and each
blocker can be shared with at most one candidate of each color other than \(\col(y)\) and
the blocker's color, giving at most \(m-2\) such candidates.

Let \(C=2^{m-1}\).  In each label class, if the class contains at most
\(C(\Delta+1)\) candidates, retain all of them.  Otherwise, retain an arbitrary subset of
\(C(\Delta+1)\) candidates from that class.  We call a label class
\emph{truncated} in the latter case.  By Lemma~\ref{lem:hidden-certificate}, for some \(x\in\Bad(T)\), the graph
\(H_x\) contains a certificate \(A^\star\) of size at most \(C\), extracted from the
optimal solution.  Starting with \(A^\star\), process its label classes one at a time and
replace its candidates by retained candidates of the same label.  Throughout this process,
the current set remains a certificate and all candidates in previously processed classes
are retained.

If a label class was not truncated, then every candidate in that class was retained, so no
replacement is needed.  Now consider a truncated class from which the current certificate
uses \(a\) candidates.  Remove these \(a\) candidates temporarily.  The current choices in
all other classes number at most \(C-a\), so they conflict with at most
\((C-a)\Delta\) of the \(C(\Delta+1)\) retained candidates in this class.  Hence at least
\[
        C(\Delta+1)-(C-a)\Delta
        =C+a\Delta
        \ge a
\]
retained candidates do not conflict with the current choices.  Choose any \(a\) of them.
They do not conflict with one another because each label class is independent.  Since they
have the same label as the candidates they replace, the certificate condition is unchanged.
After all label classes have been processed, the resulting certificate uses only retained
candidates.

The total number of retained candidates is at most
\(m2^{m-1}C(\Delta+1)=O(m^3 4^m)\).  Enumerate its subsets of size at most \(C\),
and test the certificate conditions directly.  For the point \(x\) guaranteed above, the
enumeration finds a certificate.  The running time is \(2^{O(m2^m)}\) up to polynomial
factors in \(n\).  Trying all bad points contributes another polynomial factor.
\end{proof}

\begin{proof}[Proof of Theorem~\ref{thm:main-approx}]
If \(k\le1\), any feasible set is optimal, and if \(\mathrm{OPT}=0\), any feasible set is
a valid approximation.  Guess \(\mathrm{OPT}\) by trying every pairwise distance.  For
each guess, apply Lemma~\ref{lem:profile-sparsification} and set
\(r=\mathrm{OPT}/6\).  If some color has fewer
than its quota in the reduced instance, discard this guess.  Otherwise choose any
quota-feasible set \(T\).

While \(T\) is bad, use Lemma~\ref{lem:find-exact-certificate} to find a certificate for
some bad point \(x\), and apply its exchange.  Lemma~\ref{lem:certificate-improves} shows
that the new set remains feasible and has fewer bad points.  If no certificate is found,
discard the current guess.  At most \(n\) exchanges are applied, and when the process stops,
\(\diver(T)\ge r\).

For the correct guess, the reduced instance has same-color distance at least \(2r\)
and an optimal solution of diversity greater than \(2r\).  The certificate search therefore
never fails for the correct guess, and the resulting set has diversity at least
\(\mathrm{OPT}/6\).  Trying every distance guess and performing at most \(n\) certificate
searches for each one changes the running time only by a polynomial factor.
\end{proof}

\section{General Case}
\label{sec:general-case}

We now study the lower/upper model defined in
Section~\ref{sec:problem-formulation}.  Recall that the profile
\(\sigma(p)\subseteq[m]\) of a point \(p\) is the set of colors to which it belongs.  For
\(U\subseteq X\), let \(g(U)\in\Z_{\ge0}^m\) be the color-count vector, so
\(g(U)_j=|\{p\in U:j\in\sigma(p)\}|\).
For a single point, write \(g(p)=g(\{p\})\).  Thus \(T\) is feasible if \(|T|=k\) and
\(\ell\le g(T)\le u\).  Let \(\Pi\) be the set of distinct profiles in \(X\), put
\(\tau=|\Pi|\le 2^m\), and let \(X_P=\{p\in X:\sigma(p)=P\}\) for \(P\in\Pi\).

A point contributes one unit to the cardinality and one unit to every color in its
profile.  We collect these contributions in
\(w(p)=(1,\mathbf{1}_{\sigma(p)})\in\{0,1\}^{m+1}\), and write
\(w(U)=\sum_{p\in U}w(p)\).  Thus the zeroth coordinate of \(w(U)\) is \(|U|\), while its
remaining coordinates are \(g(U)\).
For \(v\in\Z^{m+1}\), write \(v_0\) for its zeroth coordinate and \(v_{1:m}\) for the
vector consisting of its remaining \(m\) coordinates.

\begin{theorem}[General FMMD]
\label{thm:general-approx}
FMMD with overlapping colors and lower and upper bounds in general metrics admits a
deterministic \(6\)-approximation in time \(2^{2^{O(m^2)}}\poly(n)\).
\end{theorem}

As in Section~\ref{sec:exact-quotas}, we apply
Lemma~\ref{lem:profile-sparsification} and set
\(r=\mathrm{OPT}/6\).  We continue to write \(X\) for the reduced
ground set.  Thus distinct points with the same profile are at distance
at least \(2r\), and we fix a feasible set \(S\subseteq X\) with
\(\diver(S)>2r\).


\subsection{Certificates}
\label{sec:overlapping-certificates}

In the exact-quota case, the certificate restores every color count exactly.  With lower
and upper bounds, a certificate may also delete old points that are not forced blockers.
We therefore use two kinds of single moves, called atoms.  Fix a feasible set \(T\) and a
bad point \(x\in\Bad(T)\).  An addition atom adds a new point together with the removal of
its blockers, while a deletion atom removes one further old point.

For \(y\in X\setminus T\), define
\[
        B_x(y)=\{p\in T\setminus\{x\}:\dist(p,y)<r\}.
\]
These are the old points that must be removed if \(y\) is added.  We treat the addition of
\(y\), together with the removal of its blockers, as one addition atom with vector
\[
        \delta(y)=w(y)-\sum_{p\in B_x(y)}w(p).
\]
We also allow an extra deletion atom for each \(z\in T\setminus\{x\}\), with vector
\[
        \delta(z)=-w(z).
\]


For a collection \(Q\) of atoms, let \(A(Q)\subseteq X\setminus T\)
be the set of points whose addition atoms belong to \(Q\), and let
\(E(Q)\subseteq T\setminus\{x\}\) be the set of points whose deletion
atoms belong to \(Q\).  

For $A \subseteq X \setminus T$, write $B_x(A)=\bigcup_{y\in A}B_x(y)$. 
We call \(Q\) \emph{compatible} if the points of \(A(Q)\) are pairwise
at distance at least \(r\), the sets \(B_x(y)\), \(y\in A(Q)\), are
pairwise disjoint, and
\[
        E(Q)\cap B_x(A(Q))=\varnothing.
\]
Thus no point is removed both as a blocker of an addition and by an
extra deletion atom.  Extend \(\delta\) to collections of atoms by
\[
        \delta(Q)=\sum_{a\in Q}\delta(a).
\]

\begin{definition}[Certificate]
\label{def:overlapping-certificate}
A certificate for \(x\) is a compatible collection \(Q\) of atoms such that
\[
        \delta(Q)_0=1,
        \qquad
        \ell\le g(T)-g(x)+\delta(Q)_{1:m}\le u.
\]
The zeroth coordinate is the net change in cardinality, so \(\delta(Q)_0=1\) restores the total
size to \(k\) after \(x\) is removed.  The remaining coordinates are the net changes in the
color counts, so the second condition says that every resulting count lies between its lower
and upper bounds.
\end{definition}

\begin{lemma}
\label{lem:overlapping-certificate-improves}
If \(Q\) is a certificate for \(x\), then
\[
        T'=\bigl(T\setminus(B_x(A(Q))\cup E(Q)\cup\{x\})\bigr)\cup A(Q)
\]
is feasible and satisfies \(\Bad(T')\subseteq\Bad(T)\setminus\{x\}\).  In particular,
\(|\Bad(T')|<|\Bad(T)|\).
\end{lemma}

\begin{proof}
Let \(A=A(Q)\) and \(E=E(Q)\).  Since \(Q\) is compatible, no two additions share a blocker, and no
point of \(E\) is already a blocker of an addition.  Thus the forced deletions \(B_x(A)\)
and the extra deletions \(E\) are disjoint.  With \(D=B_x(A)\cup E\), we have
\[
        \delta(Q)=w(A)-w(D).
\]
The zeroth coordinate gives \(|A|=|D|+1\), so
\[
        T'=T\setminus(D\cup\{x\})\cup A
\]
has size \(k\).  Its color-count vector is
\[
        g(T')=g(T)-g(x)+\delta(Q)_{1:m},
\]
which lies between \(\ell\) and \(u\).  Thus \(T'\) is feasible.  The points of \(A\) are
pairwise at distance at least \(r\) because \(Q\) is compatible.  Every point of
\(T\setminus(D\cup\{x\})\) is also at distance at least \(r\) from every point of \(A\),
since all closer points are blockers and belong to \(D\).  The pairs among the old points
that remain are unchanged.  Therefore $\Bad(T')\subseteq \Bad(T)\setminus\{x\}$.
Since \(x\in\Bad(T)\), this containment gives \(|\Bad(T')|<|\Bad(T)|\).
\end{proof}

\subsection{A Bounded Certificate}
\label{sec:bounded-overlapping-certificate}

Suppose that \(T\) is bad, and choose an endpoint \(x\notin S\) of a bad edge.  Such an
endpoint exists because \(\diver(S)>2r\).  A blocker of a point \(y\in S\setminus T\)
cannot belong to \(S\), because it would be a distinct point of \(S\) within distance
\(r\) of \(y\).  Hence every blocker lies in
\((T\setminus S)\setminus\{x\}\).  Moreover, two points of \(S\setminus T\) cannot share
a blocker, since that would place them at distance less than \(2r\).

Let \(\Omega\) contain the addition atom for every \(y\in S\setminus T\), together with a
deletion atom for every point of
\[
        (T\setminus S)\setminus\bigl(\{x\}\cup B_x(S\setminus T)\bigr).
\]
The added points lie in \(S\), so they are pairwise more than \(2r\) apart.  Their blocker
sets are disjoint, and the extra deletion atoms were chosen outside
\(B_x(S\setminus T)\).  Thus the atoms are compatible.  Applying them after removing \(x\)
turns \(T\) exactly into \(S\), so
\[
        \delta(\Omega)=w(S)-w(T\setminus\{x\}).
\]
In particular,
\[
        \delta(\Omega)_0=1,
        \qquad
        g(T)-g(x)+\delta(\Omega)_{1:m}=g(S).
\]
Since \(S\) is feasible, \(\Omega\) is a certificate.  Our remaining task is to show that
\(\Omega\) contains a subcollection of size bounded only by \(m\) that is itself a
certificate.

An addition has at most one blocker of each profile because points with the same profile are
\(2r\)-separated.  It has no blocker with its own profile, since such a blocker would be
within distance \(r\) of a point with the same profile.  Thus \(|B_x(y)|\le\tau-1\).  In
each coordinate, the addition contributes at most \(1\), while its blockers subtract at
most \(\tau-1\).  Every coordinate of an addition vector therefore lies between
\(-(\tau-1)\) and \(1\).  A deletion vector has coordinates in \(\{-1,0\}\), so every
atom vector has \(\ell_\infty\)-norm at most \(\tau\).

Put
\[
        H=\bigl(2(m+1)\tau+1\bigr)^{m+1},
        \qquad
        C=(m+1)H.
\]

Steinitz proved the original rearrangement lemma in~\cite{Steinitz1913}.  The
\(d\tau\) bound below is due to Grinberg and
Sevast'yanov~\cite{GrinbergSevastyanov1980}.  For completeness, we include their
proof, specialized to zero-sum vectors with the \(\ell_\infty\)-norm.

\begin{lemma}[Steinitz rearrangement]
\label{lem:steinitz}
Let \(z_1,\ldots,z_N\in\mathbb R^d\) satisfy
\[
        \sum_{i=1}^N z_i=0,
        \qquad
        \|z_i\|_\infty\le \tau.
\]
Then the vectors can be ordered so that every partial sum has
\(\ell_\infty\)-norm at most \(d\tau\).
\end{lemma}

\begin{proof}
If \(N\le d\), any ordering works by the triangle inequality, so assume \(N>d\).
We repeatedly delete vectors while maintaining, on the current index set \(A\) of size
\(q\), weights \(\lambda_i\in[0,1]\) such that
\(\sum_{i\in A}\lambda_i z_i=0\) and \(\sum_{i\in A}\lambda_i=q-d\).
Initially take \(A=[N]\) and \(\lambda_i=1-d/N\).

Suppose the current index set \(A\) has size \(q>d\).  Scaling its weights by
\((q-d-1)/(q-d)\) gives a weighting with total weight \(q-d-1\) and weighted vector sum
zero.  Among all such weightings of \(A\), choose one with as few fractional entries as
possible.  At most \(d+1\)
entries are fractional: otherwise the corresponding vectors
\((1,z_i)\in\mathbb R^{d+1}\) are linearly dependent, and moving the weights along this
dependence preserves both sums until another weight reaches \(0\) or \(1\).  Consequently,
at least \(q-d-1\) weights lie in \(\{0,1\}\), and some weight is zero.  Indeed, if none
were zero, all the integral weights would equal \(1\).  Either there would be more than
\(q-d-1\) of them, or the remaining positive fractional weights would make the total
exceed \(q-d-1\).  Delete a vector of weight zero.  The remaining \(q-1\) vectors satisfy
the invariant again.

When \(d\) vectors remain, order them arbitrarily and then append the deleted vectors in
reverse deletion order.  Every prefix with index set \(A\) and \(|A|=q\ge d\) is one of
the sets encountered above.  For its associated weights,
\(\sum_{i\in A}z_i=\sum_{i\in A}(1-\lambda_i)z_i\), and the nonnegative coefficients on
the right sum to \(q-(q-d)=d\).  The prefix therefore has \(\ell_\infty\)-norm at most
\(d\tau\).  Shorter prefixes satisfy the same bound by the triangle inequality.
\end{proof}

For a finite indexed collection \(\mathcal V\) of vectors in \(\Z^{m+1}\), write
\(v(\mathcal A)=\sum_{a\in\mathcal A}v_a\) for \(\mathcal A\subseteq\mathcal V\).
Call \(\mathcal A\) sign-compatible with \(\mathcal V\) if, for every coordinate \(j\),
\[
        v(\mathcal V)_j\ge0\ \Longrightarrow\ v(\mathcal A)_j\ge0,
        \qquad
        v(\mathcal V)_j\le0\ \Longrightarrow\ v(\mathcal A)_j\le0.
\]
In particular, if \(v(\mathcal V)_j=0\), then sign compatibility requires
\(v(\mathcal A)_j=0\).
\begin{lemma}
\label{lem:sign-compatible-decomposition}
Let \(\mathcal V\) be a finite indexed collection of vectors in \(\Z^{m+1}\) with
\(\ell_\infty\)-norm at most \(\tau\).  Then \(\mathcal V\) has a partition
\[
        \mathcal V=P_1\sqcup\cdots\sqcup P_t
\]
into nonempty collections that are sign-compatible with \(\mathcal V\), each containing
at most \(H\) vectors.
\end{lemma}

\begin{proof}
Choose a partition of \(\mathcal V\) into the maximum possible number of nonempty
sign-compatible parts, and fix one part \(P\).  It cannot itself be partitioned into two
nonempty sign-compatible collections.  Set \(d=m+1\).  For
\(j\in\{0,\ldots,m\}\), let \(f_j\) be the unit vector in coordinate \(j\) of
\(\mathbb R^d\).  If \(v(\mathcal V)_j>0\), add \(v(P)_j\) copies of \(-f_j\).  If
\(v(\mathcal V)_j<0\), add \(-v(P)_j\) copies of \(f_j\).  When
\(v(\mathcal V)_j=0\), sign compatibility gives \(v(P)_j=0\).  These added vectors
cancel \(v(P)\).  Let \(N\) be the size of the
resulting zero-sum collection.  Every vector has \(\ell_\infty\)-norm at most \(\tau\).

We claim that this collection has no nonempty proper zero-sum subcollection.  Suppose one
existed.  It could not consist only of added unit vectors, because all added vectors in one
coordinate have the same sign.  Let \(P'\subseteq P\) be its original vectors.  In coordinate
\(j\), the added unit vectors have sign opposite to \(v(\mathcal V)_j\), so their sum can
cancel \(v(P')_j\) only if \(v(P')_j\) has the same sign as
\(v(\mathcal V)_j\).  When \(v(\mathcal V)_j=0\), no unit vectors were added, so
\(v(P')_j=0\).  Thus \(P'\) is
sign-compatible.  Applying the
same argument to the complementary zero-sum subcollection shows that
\(P\setminus P'\) is nonempty and sign-compatible.  This splits \(P\) into two nonempty
sign-compatible parts, contradicting the choice of the partition.

By Lemma~\ref{lem:steinitz}, choose an ordering \(v_1,\ldots,v_N\) of the augmented
collection such that every partial sum \(s_k=\sum_{h=1}^k v_h\) satisfies
\(\|s_k\|_\infty\le d\tau\).  Each \(s_k\) is an integer point of the box
\([-d\tau,d\tau]^d\), which contains
\[
        (2d\tau+1)^d=H
\]
points.  If \(N>H\), then \(s_i=s_j\) for some \(1\le i<j\le N\).  The vectors in
positions \(i+1,\ldots,j\) form a nonempty zero-sum subcollection, and it is proper because
\(i\ge1\).  This is impossible.  Therefore the augmented collection, and hence \(P\), has
at most \(H\) vectors.  The same argument applies to every part in the partition.
\end{proof}

\begin{lemma}
\label{lem:bounded-overlapping-certificate}
Every bad feasible set \(T\) has a bad point \(x\) with a certificate of size at most
\(C\).
\end{lemma}

\begin{proof}
The collection \(\Omega\) contains the addition atom for every point of \(S\setminus T\)
and a deletion atom for every point of
\((T\setminus S)\setminus(\{x\}\cup B_x(S\setminus T))\).  It is compatible, and applying
its atoms to \(T\setminus\{x\}\) gives \(S\).  Therefore
\[
        \delta(\Omega)_0=1,
        \qquad
        \ell\le g(T)-g(x)+\delta(\Omega)_{1:m}\le u.
\]
Apply Lemma~\ref{lem:sign-compatible-decomposition} to the indexed collection of atom
vectors \(\{\delta(a):a\in\Omega\}\), and write
\[
        \Omega=P_1\sqcup\cdots\sqcup P_t
\]
as a partition into nonempty subcollections of atoms.  Thus the \(P_i\) are pairwise
disjoint and cover \(\Omega\), each has at most \(H\) atoms, and each is sign-compatible.

Since \(\delta(\Omega)_0=1\), sign compatibility gives \(\delta(P_i)_0\ge0\) for every
\(i\).  These integers sum to \(1\), so exactly one part, say \(P_*\), has
\(\delta(P_*)_0=1\), and every other \(P_i\) has zeroth coordinate zero.  Put
\(b=g(T)-g(x)\), and let
\[
        J=\{j\in[m]:b_j+\delta(P_*)_j<\ell_j\}
\]
be the colors whose lower bounds are not restored by \(P_*\).

Fix \(j\in J\).  If \(\delta(\Omega)_j\le0\), then every \(P_i\) has nonpositive
\(j\)-coordinate, so
\(\delta(P_*)_j\ge\delta(\Omega)_j\).  This would give
\(b_j+\delta(P_*)_j\ge b_j+\delta(\Omega)_j\ge\ell_j\), a contradiction.  Hence
\(\delta(\Omega)_j>0\), so every \(P_i\) has nonnegative \(j\)-coordinate.
Feasibility of \(T\) gives \(b_j\ge\ell_j-1\), and the strict inequality
now forces \(b_j=\ell_j-1\) and \(\delta(P_*)_j=0\).  Since
\(b_j+\delta(\Omega)_j\ge\ell_j\), there is an index \(i(j)\ne *\) with
\(\delta(P_{i(j)})_j>0\).  Choose one such index for each \(j\in J\).

Let
\[
        \Omega'=P_*\cup\bigcup_{j\in J}P_{i(j)}.
\]
This is a union of at most \(m+1\) parts of the partition.  Every selected part other than
\(P_*\) has zeroth coordinate zero.  Therefore
\[
        |\Omega'|\le(m+1)H=C,
        \qquad
        \delta(\Omega')_0=1.
\]

Fix a color \(j\).  We first verify the lower bound.  If
\(\delta(\Omega)_j<0\), then every part has nonpositive \(j\)-coordinate.  Since
\(\Omega'\subseteq\Omega\), omitting parts can only make the total change less negative, so
\(\delta(\Omega')_j\ge\delta(\Omega)_j\).  The full exchange is feasible, so
\(b_j+\delta(\Omega')_j\ge b_j+\delta(\Omega)_j\ge\ell_j\).
Now suppose \(\delta(\Omega)_j\ge0\).  Every selected part has nonnegative
\(j\)-coordinate.  If \(P_*\) already restores the lower bound, then so does \(\Omega'\).
Otherwise \(j\in J\).  In this case we proved that \(b_j=\ell_j-1\) and included a part
\(P_{i(j)}\) with positive integer \(j\)-coordinate.  Every other helper part included in
\(\Omega'\) also has nonnegative \(j\)-coordinate, so none can undo this gain.  Hence
\(b_j+\delta(\Omega')_j\ge\ell_j\).

For the upper bound, suppose first that \(\delta(\Omega)_j\ge0\).  All parts have
nonnegative \(j\)-coordinate, so \(\delta(\Omega')_j\le\delta(\Omega)_j\).  Hence
\(b_j+\delta(\Omega')_j\le b_j+\delta(\Omega)_j\le u_j\).
If \(\delta(\Omega)_j<0\), then \(\delta(\Omega')_j\le0\), and hence
\(b_j+\delta(\Omega')_j\le b_j\le g(T)_j\le u_j\).
Thus
\[
        \ell\le b+\delta(\Omega')_{1:m}\le u.
\]

The collection \(\Omega'\) is compatible because it is a subcollection of \(\Omega\), and
\(\delta(\Omega')_0=1\).  Hence
\[
        \delta(\Omega')_0=1,
        \qquad
        \ell\le g(T)-g(x)+\delta(\Omega')_{1:m}\le u,
\]
so \(\Omega'\) is a certificate of size at most \(C\).
\end{proof}

\subsection{Sparsifying the Candidate Set and Finding a Certificate}
\label{sec:finding-overlapping-certificate}

\begin{lemma}
\label{lem:find-overlapping-certificate}
Suppose distinct points with the same profile are at distance at least \(2r\), and the
instance has an optimal solution of diversity greater than \(2r\).  Given a bad feasible
set \(T\), one can find a bad point \(x\) and a certificate for \(x\) in time
\(2^{2^{O(m^2)}}\poly(n)\).
Without these assumptions, the same search may report \(\mathrm{FAIL}\).
\end{lemma}

\begin{proof}
For each \(x\in\Bad(T)\), let \(H_x\) be the conflict graph on the addition and deletion
atoms for \(x\).  Two additions \(y,y'\) are adjacent if
\[
        \dist(y,y')<r
        \qquad
        \text{or}
        \qquad
        B_x(y)\cap B_x(y')\ne\varnothing.
\]
An addition atom \(y\) and a deletion atom \(z\) are adjacent if \(z\in B_x(y)\).
Deletion atoms are not adjacent to one another.  Independent sets in \(H_x\) are exactly
the compatible atom collections.

Define the \emph{type} of an addition atom \(y\) as \((+,\sigma(y),\delta(y))\), and the type of a deletion atom
\(z\) as \((-,\sigma(z))\).  The type determines the atom's vector.  The blocker set of
an addition contains at most one point of each profile and none with the profile of the
addition.  Hence every coordinate of \(\delta(y)\) is an integer between
\(-(\tau-1)\) and \(1\).  There are therefore at most
\[
        L=\tau(\tau+1)^{m+1}+\tau
\]
types.  Each type is independent in \(H_x\).  Indeed, two additions of the same type have
the same profile and are therefore at distance at least \(2r\).  They cannot share a
blocker, since that would put them at distance less than \(2r\).  Deletion atoms are never
adjacent.

We next bound the maximum degree.  For each profile different from \(\sigma(y)\), at most
one addition atom \(y'\) satisfies \(\dist(y,y')<r\).  Thus there are at most \(\tau-1\)
such atoms \(y'\).  Moreover,
\(|B_x(y)|\le\tau-1\).  For each \(p\in B_x(y)\), at most \(\tau-2\) other additions
have \(p\) in their blocker set, and the deletion atom corresponding to \(p\) is also
adjacent to \(y\).  Thus
\[
        \deg_{H_x}(y)
        \le (\tau-1)+(\tau-1)^2
        \le \tau^2.
\]
If a deletion atom corresponds to \(p\), its neighbors are precisely the addition atoms
\(y\) for which \(p\in B_x(y)\).  There is at most one such atom of each profile other
than \(\sigma(p)\), so its degree is at most \(\tau-1\).  Therefore the maximum degree
\(\Delta=\Delta(H_x)\) satisfies \(\Delta\le\tau^2\).

In every graph \(H_x\), consider each type separately.  If a type
contains at most \(C(\Delta+1)\) atoms, retain all of them.  Otherwise,
retain an arbitrary \(C(\Delta+1)\) atoms of that type.  We call a type
\emph{truncated} in the latter case.

Lemma~\ref{lem:bounded-overlapping-certificate} guarantees that, for
some \(x\in\Bad(T)\), the full graph \(H_x\) contains a certificate
\(Q^\star\) of size at most \(C\).  This certificate may use atoms that
were not retained.  We show that we can replace its atoms one type at
a time to obtain a certificate using only retained atoms.

Start with \(Q^\star\) and process the types one at a time.  For an
untruncated type, all atoms of that type were retained, so there is
nothing to replace.  Now consider a truncated type that occurs \(a\)
times in the current certificate.  Temporarily remove these \(a\)
atoms.  At most \(C-a\) atoms remain, and they conflict with at most
\((C-a)\Delta\) retained atoms of the
current type.  Since we retained \(C(\Delta+1)\) atoms of this type, at
least
\[
        C(\Delta+1)-(C-a)\Delta
        =C+a\Delta
        \ge a
\]
of them do not conflict with the remaining atoms.  Choose any \(a\) of
these atoms.  They do not conflict with one another because atoms of
the same type form an independent set.  We can therefore use them in
place of the \(a\) removed atoms without creating any conflicts.

Repeat this replacement for every type.  The final collection \(Q\)
uses only retained atoms and has the same number of atoms of each type
as \(Q^\star\).  Since an atom's type determines its vector, we have
\(\delta(Q)=\delta(Q^\star)\).  Thus \(Q\) is compatible and satisfies
the same certificate conditions as \(Q^\star\), so \(Q\) is also a
certificate.

Since \(\tau\le2^m\), we have \(C=2^{O(m^2)}\), \(L=2^{O(m^2)}\), and
\(\Delta=2^{O(m)}\).  Each retained graph therefore has at most $N=LC(\Delta+1)=2^{O(m^2)}$ vertices.  Enumerate every subset of size at most \(C\), and test compatibility and the
certificate conditions directly.  This takes $N^{O(C)}=2^{2^{O(m^2)}}$ time.  For the point \(x\) supplied by
Lemma~\ref{lem:bounded-overlapping-certificate}, the enumeration finds a certificate.
Trying every bad point adds only a polynomial factor.
\end{proof}

\subsection{Finding a Feasible Solution}

In the exact-quota case, a feasible starting set is obtained by choosing \(k_i\) points
of each color \(i\).  With overlapping colors, finding a feasible starting set is a little trickier because selecting one point may change several
color counts. The lemma below helps us find this feasible starting set.

\begin{lemma}
\label{lem:find-feasible-general}
Given a set $Y\subseteq X$, one can find a feasible set $T\subseteq Y$, or determine that none exists, in time
\[
2^{2^{O(m)}}\poly(n).
\]
\end{lemma}

\begin{proof}
For each profile $P\in\Pi$, let $n_P=|Y\cap X_P|$, and introduce an integer
variable $c_P$ denoting the number of selected points of profile $P$. A feasible
set exists if and only if
\[
0\le c_P\le n_P\quad(P\in\Pi),\qquad
\sum_{P\in\Pi}c_P=k,\qquad
\ell_j\le\sum_{P\ni j}c_P\le u_j\quad(j\in[m]).
\]
Indeed, a feasible set gives such integers $c_P$, and conversely any integer
solution can be realized by choosing arbitrarily $c_P$ points from $Y\cap X_P$
for each $P$.

This integer program has $d=|\Pi|\le 2^m$ variables. By Lenstra's
fixed-dimensional integer-programming algorithm~\cite{Lenstra1983}, with the
quantitative bound of Kannan~\cite{Kannan1987}, it can be solved in time
$d^{O(d)}\poly(L)$, where $L$ is the input length in bits. Here all coefficients
are $0$ or $1$, all numerical bounds are at most $n$, and there are at most
$2^m$ variables and $O(2^m+m)$ constraints, so $L=2^{O(m)}\log n$. Hence the
running time is
\[
(2^m)^{O(2^m)}\poly(L)
=2^{O(m2^m)}\poly(n)
=2^{2^{O(m)}}\poly(n).
\]
\end{proof}

\begin{proof}[Proof of Theorem~\ref{thm:general-approx}]
The case \(k\le1\) can be solved directly, so assume \(k\ge2\).  Guess \(\mathrm{OPT}\) by
trying every pairwise distance.  For each guess, apply
Lemma~\ref{lem:profile-sparsification} to obtain \(Y\), and set \(r=\mathrm{OPT}/6\).
Use Lemma~\ref{lem:find-feasible-general} to find a feasible set \(T\subseteq Y\).  If no
such set exists, discard the guess.

While \(T\) is bad, use Lemma~\ref{lem:find-overlapping-certificate} to find a certificate
for some bad point \(x\), and apply its exchange.  By
Lemma~\ref{lem:overlapping-certificate-improves}, the new set remains feasible and has fewer
bad points.  If no certificate is found, discard the current guess.  After at most \(n\) iterations, the process stops with \(\diver(T)\ge r\).

For the correct guess, Lemma~\ref{lem:profile-sparsification} guarantees that distinct
points of \(Y\) with the same profile are at distance at least \(2r\), and that \(Y\)
contains a feasible solution of diversity greater than \(2r\).  Thus
Lemma~\ref{lem:find-feasible-general} finds an initial feasible set, and the certificate
search never fails.  The resulting set has diversity at least \(\mathrm{OPT}/6\).
Lemma~\ref{lem:find-feasible-general} takes \(2^{2^{O(m)}}\poly(n)\) time, which is
smaller than the \(2^{2^{O(m^2)}}\poly(n)\) time used by the certificate search.  Trying
every distance guess and applying at most \(n\) certificates per guess changes the running
time only by a polynomial factor.
\end{proof}

\section{Discussion}

The main open question is whether exact-quota FMMD admits an $O(1)$-approximation
in time $\poly(n,m)$. 
The $2^{m-1}$ certificate bound for exact quotas that we use is tight. Thus the exponential
dependence on $m$ in our bounded-swap approach cannot be removed simply by
proving the existence of smaller certificates. Obtaining a polynomial-time
algorithm appears to require a different way of finding or representing large
exchanges, or a different framework.

At the other
extreme, it would be interesting to prove that no constant-factor approximation
for exact-quota FMMD is possible in polynomial time, so that the approximation
ratio must necessarily grow with $m$. It would also be interesting to improve
the dependence on $m$ in our fixed-parameter algorithms, or to improve the
approximation factor $6$ while retaining exact feasibility.

\section*{AI Disclosure}
We developed the local-search algorithm, identified the number of bad points as the appropriate potential function, and proved the existence of improving repairs of size at most \(h(m)\). These ideas, developed without any AI assistance, yielded an algorithm with running time \((mk)^{h(m)}\). After being provided with our draft, GPT-5.6 contributed the idea of using type-based sparsification of the candidate set to improve the running time to \(f(m)\poly(n)\) and also assisted in extending the approach to the general color-constraint setting. We independently verified all AI-assisted arguments and take full responsibility for the correctness and originality of the paper.
\bibliographystyle{alpha}
\bibliography{references}

\end{document}